\documentclass[a4paper,onecolumn,11pt,accepted=2025-11-12]{quantumarticle}
\pdfoutput=1
\usepackage[utf8]{inputenc}
\usepackage[english]{babel}
\usepackage[T1]{fontenc}
\usepackage{amsmath}
\usepackage{hyperref}

\usepackage{tikz}
\usepackage{lipsum}

\usepackage{amsmath,amssymb,amsthm,mathrsfs}
\usepackage{epsfig,epsf,subfigure,graphicx,graphics}
\usepackage{url,enumerate}
\usepackage{physics}
\usepackage{hyperref}

\DeclareMathAlphabet\mathrsfso      {U}{rsfso}{m}{n}
\usepackage{amssymb}\usepackage{mathtools}
\usepackage[toc,page]{appendix}
\usepackage{bbm}

\usepackage{xcolor}

\newtheorem{lemma}{Lemma}
\newtheorem{definition}{Definition}
\newtheorem{proposition}{Proposition}
\newtheorem{theorem}{Theorem}
\newtheorem{corollary}{Corollary}

\newtheorem*{proposition*}{Proposition 3}
\newtheorem*{lemma*}{Lemma 2}
\newtheorem*{lemma**}{Lemma 3}
\newtheorem*{lemma***}{Lemma 4}
\newtheorem*{theorem*}{Theorem 3}

\usepackage{algorithm}
\usepackage{algpseudocode}
\usepackage{algpseudocode}

\errorcontextlines\maxdimen
\makeatletter
\newcommand*{\algrule}[1][\algorithmicindent]{\makebox[#1][l]{\hspace*{.5em}\vrule height .75\baselineskip depth .25\baselineskip}}
\newcount\ALG@printindent@tempcnta
\def\ALG@printindent{
    \ifnum \theALG@nested>0
        \ifx\ALG@text\ALG@x@notext
            \addvspace{-3pt}
        \else
            \unskip
            \ALG@printindent@tempcnta=1
            \loop
                \algrule[\csname ALG@ind@\the\ALG@printindent@tempcnta\endcsname]
                \advance \ALG@printindent@tempcnta 1
            \ifnum \ALG@printindent@tempcnta<\numexpr\theALG@nested+1\relax
            \repeat
        \fi
    \fi
    }
\usepackage{etoolbox}
\patchcmd{\ALG@doentity}{\noindent\hskip\ALG@tlm}{\ALG@printindent}{}{\errmessage{failed to patch}}
\makeatother

\begin{document}

\title{Bounds in Sequential Unambiguous Discrimination of Multiple Pure Quantum States}

\author{Jordi Pérez-Guijarro}
\affiliation{SPCOM Group, Universitat Politècnica de Catalunya, Barcelona, Spain}

\author{Alba Pagès-Zamora}
\affiliation{SPCOM Group, Universitat Politècnica de Catalunya, Barcelona, Spain}

\author{Javier R. Fonollosa}
\affiliation{SPCOM Group, Universitat Politècnica de Catalunya, Barcelona, Spain}

\maketitle

\maketitle

\begin{abstract}
   Sequential methods for quantum hypothesis testing offer significant advantages over fixed-length approaches, which rely on a predefined number of state copies. Despite their potential, these methods remain underexplored for unambiguous discrimination. In this work, we derive performance bounds for such methods when applied to the discrimination of a set of pure states. The performance is evaluated based on the expected number of copies required. We establish a lower bound applicable to any sequential method and an upper bound on the optimal sequential method. The upper bound is derived using a novel and simple non-adaptive method. Importantly, the gap between these bounds is minimal, scaling logarithmically with the number of distinct states.
\end{abstract}

The problem of identifying a discrete set of states is central to quantum information science. Insights from this task have far-reaching implications, including setting theoretical bounds on quantum error mitigation (QEM) \cite{takagi2022fundamental}, as well as advancements in fields such as quantum cryptography \cite{acin2006secrecy} and quantum communication \cite{gisin2007quantum}.

As this task is so fundamental, it has a rich history of study and results. The study of this task began with binary state discrimination, that is, distinguishing between two states, $\sigma_1$ and $\sigma_2$. The optimal measurement is given by the Helstrom measurement \cite{Helstrom}, which yields a probability of error given by $ \frac{1}{2} \left( 1 - \left\| \pi_1 \sigma_1 - \pi_2 \sigma_2 \right\|_1 \right)$, where $\pi_i$ denotes the a priori probability of hypothesis $i$. With more copies of the states, the probability of error decreases exponentially, motivating the study and characterization of the decay rate of the probability of error. The optimal decay rate is given by the quantum Chernoff distance \cite{Chernoff_1, Chernoff_2}, which generalizes its classical counterpart. For the more general case where $N \geq 2$ states are considered, the optimal decay rate is also well-characterized and given by $\min_{i,j: i \neq j} C_Q(\sigma_i, \sigma_j)$ \cite{Multiple}, where $C_Q(\sigma, \rho)$ is the Chernoff distance between states $\sigma $ and $\rho$.

A variant of this problem is the unambiguous discrimination of states. In this setting, any decision made must be correct. To guarantee this requirement, the option of choosing not to decide is allowed. That is, for $N$ states, there are $N+1$ possible results of the discrimination process, with one completely uninformative. However, not every set of states allows for unambiguous discrimination, as certain constraints must be met. For example, for pure states and non-adaptive methods, the states must be linearly independent. This problem has been widely studied for the cases of $N=2$ \cite{herzog2005optimum, jaeger1995optimal, ivanovic1987differentiate, dieks1988overlap, peres1988differentiate} and $N=3$ \cite{bergou2012optimal, sentis2022online}. In contrast, for $N > 3$, finding an analytical solution for the procedure that minimizes the probability of an inconclusive answer becomes extremely complex.

The results discussed so far are about fixed-length methods, which involve a fixed number of state copies. Yet, other approaches, known as sequential methods, proposed for the task of quantum hypothesis testing in \cite{slussarenko2017quantum}, are also studied for both quantum hypothesis testing and unambiguous discrimination. Sequential methods, unlike fixed-length methods, deal with a random number of state copies. At the cost of this uncertainty, they offer improved performance \cite{martinez2021quantum,li2022optimal,perez2022quantum}. Interestingly, for the case of $N$ pure states, as shown in \cite{perez2022quantum}, using only simple non-collective measurements, we can achieve unambiguous discrimination using an expected number of copies that scales as $O(N)$. Much like fixed-length methods, the problem of unambiguous discrimination has been scarcely studied for the general case of $N$ states using sequential methods.

In this paper, we analyze the theoretical limits on the expected number of state copies required by sequential methods for the unambiguous discrimination of $N$ pure states. First, we derive a lower bound on the number of copies needed by any sequential method that performs unambiguous discrimination. Next, we establish an upper bound on the performance of the optimal sequential method. The gap between these bounds is of the order $O(\log N)$, which is relatively small.

\section{Quantum Hypothesis Testing and Unambiguous Discrimination}

Quantum hypothesis testing involves distinguishing between different quantum states or hypotheses based on measurement outcomes. Specifically, given multiple copies of an unknown state $\sigma$, the task is to determine which state from the set $\{\sigma_i\}_{i=1}^N$ it corresponds to. Here, we analyze a special case where all the states are pure, i.e., $\sigma_i = \ket{\psi_i}\bra{\psi_i}$ for every $i \in [N]$,  where $[N] := \{1, 2, \dots, N\}$. The decision $D \in [N]$ is a random variable (r.v.) that depends on the outcome of the measurements performed on the copies of state $\sigma$. The true hypothesis is denoted by $S$, i.e., $\sigma = \sigma_S$. Importantly, $S\in[N]$ is also a r.v., with probabilities given by the priors $\{\pi_i\}_{i=1}^N$.

Therefore, the probability of error is given by $\mathbb{P}(D \neq S)$ and for simplicity, it is denoted by $\mathbb{P}(\mathcal{E})$. Note that in fixed-length methods, the number of copies used to make a decision, denoted by $L$, is preassigned. However, in sequential methods, where measurements are performed in an online fashion, the decision is not made until some stopping criterion is met. Consequently, the number of copies may vary depending on the measurement results, so $L$ becomes a random variable.

In unambiguous discrimination, the decision must be either $D = S$, meaning the true hypothesis is identified, or $D = N+1$, which accounts for cases where the discrimination fails. The value $N+1$ is arbitrary and merely represents this additional outcome associated with the inconclusive response event, denoted by $\mathcal{I}$. For fixed-length methods, the inconclusive probability must always be strictly positive, whereas for sequential methods, it can also be exactly zero. Additionally, sequential methods have another advantage over fixed-length methods. No necessary condition is required for the existence of a sequential method that solves the unambiguous discrimination problem, whereas for fixed-length methods, it is necessary that the states $\{\ket{\psi_i}^{\otimes L}\}_{i=1}^N$ are linearly independent \cite{chefles1998unambiguous}. Throughout the paper, we use the shorthand $\mathbb{P}(\mathcal{I}|s)$ for $\mathbb{P}(\mathcal{I}|S=s)$, and $\mathbb{E}[L|s]$ for $\mathbb{E}[L|S=s]$.

Finally, for the sake of clarity, we introduce the formal definition of a sequential method for quantum unambiguous discrimination.

\begin{definition} A sequential method for quantum unambiguous discrimination is a method that can be expressed in the form of Algorithm \ref{alg_1}, and furthermore satisfies $\mathbb{P}(D\neq S)=0$.

\begin{algorithm}[H]
 \caption{}
 \label{alg_1}
 \begin{algorithmic}
 	\State \textbf{Input:} Access to a query function that outputs copies of the unknown state $\ket{\psi}$.
 	\State \textbf{Output:} Decision $d \in [N+1]$.

 	\State $i \gets 0$

 	\While{stopping criterion is not satisfied}
 		\State $\bullet$ Update the counter: $i \gets i + 1$
 		\State $\bullet$ Request $k_i$ copies of the unknown state $\ket{\psi}$, where $k_i$ may depend on the data from previous iterations.
 		\State $\bullet$ Measure state $\ket{\psi}\bra{\psi}^{\otimes k_i} \otimes \,\tau_{i-1}$ using measurement $\mathcal{L}_i$, obtaining outcome $x_i$, where $\tau_{i-1}$ denotes the collection of all post-measurement states from previous iterations. This detail is included to account for measurements that do not completely collapse the state, such as weak measurements. Similar to $k_i$, the measurement $ \mathcal{L}_i$ may depend on \( (\{x_j\}_{j=1}^{i-1}, \{k_j\}_{j=1}^{i-1}, \{\mathcal{L}_j\}_{j=1}^{i-1}) \).
 		\State $\bullet$ Check whether the stopping criterion is satisfied.
 	\EndWhile

        \State Take a decision $d\in [N+1]$, which is a function of the triplet $(\{x_j\}_{j=1}^{i},\{k_j\}_{j=1}^{i},\{\mathcal{L}_j\}_{j=1}^{i})$.
 \end{algorithmic}
\end{algorithm}

\end{definition}

Interestingly, aside from the distinction between fixed-length and sequential methods, additional distinctions can be made to further characterize these methods based on their properties. For instance, if the measurements used do not depend on previous outcomes, the method is said to be non-adaptive; otherwise, it is adaptive.

\section{Lower bounds on the expected number of copies}

In this section, we first present a lower bound on the expected number of copies used by any sequential method that unambiguously discriminates a set of pure states. Next, we apply the same technique to obtain an analogous bound for quantum hypothesis testing. Specifically, we adopt a worst-case approach, where we bound the expected value $\max_{s \in [N]} \mathbb{E}[L | s]$ rather than $\mathbb{E}[L]$. This approach allows for a more straightforward technical analysis and interpretation.

\subsection{Lower Bound in unambiguous discrimination}

We start by presenting the main result of this section, which establishes a bound on the expected number of copies required for unambiguous discrimination, and then proceed with its proof. Importantly, the results obtained in this section do not require any particular restriction on the sequential methods used or on the states to be discriminated.

\begin{theorem}\label{theorem_general_lower_bound}
    For any sequential method that distinguishes unambiguously a set of pure states $\{\ket{\psi_i}\}_{i=1}^N$, the maximum expected value of state copies $\max_{s\in[N]}\mathbb{E}[L|s]$ satisfies
    %---------------
    \begin{equation}\label{theorem_lower_bound_inconclusive_prop}
        \frac{1-\max_{s\in [N]} \mathbb{P}(\mathcal{I}|s)}{-\log \max_{i,j: i\neq j} |\bra{\psi_i}\ket{\psi_j}|^{2}}\left( \frac{\log 3}{3 }\right) \leq \left \lceil \max_{s\in[N]} \mathbb{E}[L|s]\right \rceil
    \end{equation}
    %---------------
\end{theorem}
%--------------
The most interesting aspect of this lower bound is that it is inversely proportional to the minimum pairwise Chernoff distance. This is because the Chernoff distance can be expressed as $ -\log \text{Tr}(\sigma_1 \sigma_2)$ when at least one of the two states is pure \cite{casa_2}. That is, the geometry of the states has a significant impact on the performance of the methods. Another interesting point about the bound is its dependence on the inconclusive probability, where a higher probability reduces the required number of copies, as expected.

To prove this theorem, we first need to derive some auxiliary results. In particular, we use Corollary \ref{corollary_probability_of_error} and Lemma \ref{ref_sequential_to_fixed} to derive the bound for $\mathbb{P}(\mathcal{I}|s)=0$. Next, using the result for $\mathbb{P}(\mathcal{I}|s)=0$ and Lemma \ref{lemma_3}, the final result is obtained. We begin by deriving Corollary \ref{corollary_probability_of_error}, which is a specific case of the following lemma.

\begin{lemma}\label{first_lemma}
    Let $ \mathbb{P}_{\mathrm{succ}}(\{\sigma_i\}_{i=1}^N) $ be the optimal probability of success in discriminating the states $\{\sigma_i\}_{i=1}^N$ given a single copy with prior probability of $ 1/N $ for all the states. Then, for any subset $ \mathcal{S} \subseteq [N] $,
    %------------------
    \begin{equation}
       \mathbb{P}_{\mathrm{succ}}( \{\sigma_i\}_{i=1}^N)\leq \frac{|\mathcal{S}|}{N}\mathbb{P}_{\mathrm{succ}}( \{\sigma_i\}_{i\in \mathcal{S}})+\frac{|\overline{\mathcal{S}}|}{N}\mathbb{P}_{\mathrm{succ}}(  \{\sigma_i\}_{i\in \overline{\mathcal{S}}})
    \end{equation}
    %---------------
    where $\overline{\mathcal{S}}=[N]\backslash \mathcal{S}$. 
\end{lemma}

\begin{proof}
    For this proof, we need to express the optimal probability of success as the probability of success given by the optimal positive operator-valued measure (POVM). That is,
    %-------------
    \begin{align}
        \mathbb{P}_{\mathrm{succ}}(\{\sigma_i\}_{i=1}^N)&=\frac{1}{N} \sup_{ \substack{\{\Lambda_i\succeq 0\}_{i=1}^N\\\mathrm{s.t.} \sum_{i=1}^N \Lambda_i=I}} \sum_{i=1}^N \Tr(\Lambda_i \sigma_i)\nonumber \\ &= \sup_{\substack{0\preceq\Theta_1,\Theta_2\preceq I\\ \Theta_1+\Theta_2=I }}\left \{\frac{1}{N} \sup_{ \substack{\{\Lambda_i\succeq 0\}_{i\in \mathcal{S}}\\\mathrm{s.t.} \sum_{i\in \mathcal{S}} \Lambda_i=\Theta_1}} \sum_{i\in \mathcal{S}}\Tr(\Lambda_i \sigma_i)+\frac{1}{N} \sup_{ \substack{\{\Lambda_i\succeq 0\}_{i\in \overline{\mathcal{S}}}\\\mathrm{s.t.} \sum_{i\in \overline{\mathcal{S}}} \Lambda_i=\Theta_2}} \sum_{i\in \overline{\mathcal{S}}} \Tr(\Lambda_i \sigma_i) \right \} \nonumber 
    \end{align}
    \begin{align}    
        \hspace{0.5cm} & \leq \frac{1}{N}\sup_{ \substack{\{\Lambda_i\succeq 0\}_{i\in \mathcal{S}}\\\mathrm{s.t.} \sum_{i\in \mathcal{S}} \Lambda_i\preceq I}}\sum_{i\in \mathcal{S}} \Tr(\Lambda_i \sigma_i)+\frac{1}{N}\sup_{ \substack{\{\Lambda_i\succeq 0\}_{i\in \overline{\mathcal{S}}}\\\mathrm{s.t.} \sum_{i\in \overline{\mathcal{S}}} \Lambda_i\preceq I}} \sum_{i\in \overline{\mathcal{S}}} \Tr(\Lambda_i \sigma_i) \nonumber \\ &= \frac{1}{N}\sup_{ \substack{\{\Lambda_i\succeq 0\}_{i\in \mathcal{S}}\\\mathrm{s.t.} \sum_{i\in \mathcal{S}} \Lambda_i= I}}\sum_{i\in \mathcal{S}} \Tr(\Lambda_i \sigma_i)+\frac{1}{N}\sup_{ \substack{\{\Lambda_i\succeq 0\}_{i\in \overline{\mathcal{S}}}\\\mathrm{s.t.} \sum_{i\in \overline{\mathcal{S}}} \Lambda_i=I}} \sum_{i\in \overline{\mathcal{S}}} \Tr(\Lambda_i \sigma_i) \nonumber \\ &=\frac{|\mathcal{S}|}{N}\mathbb{P}_{\mathrm{succ}}( \{\sigma_i\}_{i\in \mathcal{S}})+\frac{|\overline{\mathcal{S}}|}{N}\mathbb{P}_{\mathrm{succ}}(  \{\sigma_i\}_{i\in \overline{\mathcal{S}}})
    \end{align}
    %--------------
    where the notation $A\succeq 0$ means that $A$ is positive semi-definite.  The first equality uses that $0\preceq\sum_{i\in S} \Lambda_i 	\preceq I$, and $0\preceq\sum_{i\in \overline{S}} \Lambda_i 	\preceq I$, for any POVM, and in particular, for the optimal one. The inequality follows from removing the constraint $\Theta_1+\Theta_2=I$, meaning we allow a broader set of matrices $\Lambda_i$. Consequently, this modification can only increase the value or leave it unchanged. Finally, the third equality follows from the fact that the supremum is always achieved when $\sum_{i \in S} \Lambda_i = I$.
\end{proof}
%-----------------
From this result, the corollary used in the proof of Theorem \ref{theorem_general_lower_bound} follows.

\begin{corollary}\label{corollary_probability_of_error}

Let $\mathbb{P}_{\mathrm{error}}(\{\sigma_i\}_{i=1}^N)=1-\mathbb{P}_{\mathrm{succ}}(\{\sigma_i\}_{i=1}^N)$ be the optimal probability of error in discriminating the states $\{\sigma_i\}_{i=1}^N$ given a single copy with prior probability of $1/N$ for all the states. Then,
    \begin{equation}\label{eq_probability_of_error}
        \mathbb{P}_{\mathrm{error}}(\{\sigma_i\}_{i=1}^N)\geq \frac{1}{N} \left(1-\frac{1}{2} \min_{i\neq j}\|\sigma_i-\sigma_j\|_1\right)
    \end{equation}
    %-----------------
    implying that for pure states, 
    %-----------------
    \begin{equation}
        \mathbb{P}_{\mathrm{error}}(\{\sigma_i\}_{i=1}^N)\geq \frac{1}{N} \left(1-\sqrt{1- \max_{i\neq j} |\bra{\psi_i}\ket{\psi_j}|^2} \right)
    \end{equation}

\end{corollary}

\begin{proof}

    The proof follows trivially by applying Lemma \ref{first_lemma} and performing some algebraic manipulations.
    %---------------
    \begin{align}
        \mathbb{P}_{\mathrm{error}}(\{\sigma_i\}_{i=1}^N)&=1-\mathbb{P}_{\mathrm{succ}}(\{\sigma_i\}_{i=1}^N) \nonumber \\ &\geq 1- \frac{2}{N} \, \mathbb{P}_{\mathrm{succ}}(\{\sigma_i\}_{i\in \{p,q\}})-\frac{N-2}{N} \nonumber \\ &=\frac{2}{N} \left( 1- \mathbb{P}_{\mathrm{succ}}(\{\sigma_i\}_{i\in \{p,q\}})  \right) \nonumber  \\ &= \frac{2}{N} \left( 1- \frac{1}{2}- \frac{1}{4}  \|\sigma_p-\sigma_q\|_1 \right) \nonumber \\ & = \frac{1}{N} \left( 1- \frac{1}{2}\|\sigma_p-\sigma_q\|_1 \right) 
    \end{align}
    %-------------
    where the inequality follows from Lemma \ref{first_lemma}, with $\mathcal{S} = \{p, q\}$, and from the fact that $\mathbb{P}_{\mathrm{succ}}(\{\sigma_i\}_{i \in \overline{\mathcal{S}}})$ is upper bounded by $1$. The final result is obtained by taking $\{p,q\}$ such that $\|\sigma_p-\sigma_q\|_1=\min_{i\neq j}\|\sigma_i-\sigma_j\|_1$.
\end{proof}

Next, we state the second auxiliary result in Lemma \ref{ref_sequential_to_fixed}.

\begin{lemma}\label{ref_sequential_to_fixed}
    For any sequential method that unambiguously distinguishes a set of pure states $ \{ \ket{\psi_i} \}_{i=1}^N $ with $\mathbb{P}(\mathcal{I})=0$ using a maximum expected number of measurements $\max_{s\in [N]} \mathbb{E}[L|s] $, there exists at least one fixed-length method that uses $ AK  \left \lceil \max_{s\in[N]} \mathbb{E}[L|s]\right \rceil $ samples and distinguishes the states with a probability of error at most $ 1/K^A $, where $A,K\in \mathbb{N}$.
\end{lemma}

\begin{proof}
    See Appendix \ref{Appendix_B}.
\end{proof}

Now that Corollary \ref{corollary_probability_of_error} and Lemma \ref{ref_sequential_to_fixed} have been presented, we can prove  Theorem \ref{theorem_general_lower_bound} for $\mathbb{P}(\mathcal{I}|s)=0$. To do so, we use the fact that the specific fixed-length method used in Lemma \ref{ref_sequential_to_fixed} satisfies that its error probability is lower bounded by that of the optimal fixed-length method. Additionally, using Corollary \ref{corollary_probability_of_error}, we obtain the following result.
%-------------
\begin{align}\label{expression_p_error}
    \frac{1}{K^A} &\geq  \mathbb{P}_{\mathrm{error}}\left(\left\{\ket{\psi_i}\bra{\psi_i}^{\otimes AK \left \lceil \max_{s\in[N]} \mathbb{E}[L|s]\right \rceil} \right\}_{i=1}^N \right) \nonumber \\&\geq \frac{1}{N} \left(1-\sqrt{1- \max_{i\neq j} |\bra{\psi_i}\ket{\psi_j}|^{2AK\left \lceil \max_{s\in[N]} \mathbb{E}[L|s]\right \rceil}} \right)
\end{align}
%-----------
for any $A\geq 1$. The power $2AK \left \lceil \max_{s\in[N]} \mathbb{E}[L|s]\right \rceil$ appears due to the property $\bra{\psi_i}^{\otimes u} \ket{\psi_j}^{\otimes u} = \bra{\psi_i} \ket{\psi_j}^u$, as we apply Corollary \ref{corollary_probability_of_error} to states $\left\{\ket{\psi_i}^{\otimes AK\left \lceil \max_{s\in[N]} \mathbb{E}[L|s]\right \rceil} \right\}_{i=1}^N$. Therefore, 
%-----------
\begin{equation}
    1-\frac{N}{K^A} \leq  \sqrt{1- \max_{i\neq j} |\bra{\psi_i}\ket{\psi_j}|^{2AK\left \lceil \max_{s\in[N]} \mathbb{E}[L|s]\right \rceil}}
\end{equation}
%------------
which implies that for any $A$ such that $(1 - N/K^A) \geq  0$, i.e., $A\geq \log_K N$,
%------------
\begin{equation}
    1-\left(1-\frac{N}{K^A} \right)^2 \geq  \max_{i\neq j} |\bra{\psi_i}\ket{\psi_j}|^{2AK\left \lceil \max_{s\in[N]} \mathbb{E}[L|s]\right \rceil} 
\end{equation}
%-----------
Next, using that $\log \max_{i\neq j} |\bra{\psi_i}\ket{\psi_j}|^{2}<0$,
%-----------
\begin{align}
  \left \lceil \max_{s\in[N]} \mathbb{E}[L|s]\right \rceil &\geq  \frac{\log\left(1-\left(1-\frac{N}{K^A} \right)^2\right)}{AK\log \max_{i\neq j} |\bra{\psi_i}\ket{\psi_j}|^{2}} \nonumber \\ & = \frac{1}{-\log \max_{i\neq j} |\bra{\psi_i}\ket{\psi_j}|^{2}}\left(\frac{\log K}{K}\right)\left( \frac{\delta}{(1+\delta) }-\frac{\log \left(2-\frac{1}{N^{\delta}}\right)}{(1+\delta) \log N}\right) \nonumber\\ &\geq \frac{1}{-\log \max_{i\neq j} |\bra{\psi_i}\ket{\psi_j}|^{2}} \left(\frac{\log K}{K}\right) \left( \frac{\delta-1}{1+\delta  }\right)
\end{align}
%-----------
where the equality follows by substituting $A=\log_K\left( N^{1+\delta}\right)$, being $\delta \in \mathbb{R}^+$. The second inequality uses that $N\geq 2$ and, therefore, $\log(2-\frac{1}{N^\delta})/\log N\leq 1$. Finally, by taking the supremum over $\delta \in \mathbb{R}^+$ and $K\in \mathbb{N}$,
%-----------
\begin{equation}\label{inequality_zero_inconclusive}
    \frac{1}{-\log \max_{i\neq j} |\bra{\psi_i}\ket{\psi_j}|^{2}}\left(\frac{\log 3}{3}\right) \leq \left \lceil \max_{s\in[N]} \mathbb{E}[L|s]\right \rceil
\end{equation}
%---------------
for any sequential method that discriminates unambiguously a set of quantum states $\{\ket{\psi_i}\}_{i=1}^N$ with inconclusive probability $\mathbb{P}(\mathcal{I})$=0.  Finally, we use the following lemma to conclude the proof.
%----------------
\begin{lemma}\label{lemma_3}
    For any sequential method that unambiguously discriminates a set of states $\{\ket{\psi_i}\}_{i=1}^N$ with an inconclusive probability $\mathbb{P}(\mathcal{I}|s)<1 $ and an expected number of copies $\mathbb{E}[L|s]$, there exists another sequential method that unambiguously discriminates the same set of states $\{\ket{\psi_i}\}_{i=1}^N$ with an inconclusive probability $\mathbb{P}(\mathcal{I}|s) = 0$ and an expected number of copies equal to $\mathbb{E}[L|s]/(1-\mathbb{P}(\mathcal{I}|s))$.
\end{lemma}
%-------------
\begin{proof}
    See Appendix \ref{appendix_A}.
\end{proof}
%--------------
That is, we can transform a sequential method that unambiguously discriminates a set of states with a strictly positive inconclusive probability into one with zero inconclusive probability by increasing the expected number of samples by a factor of $(1-\mathbb{P}(\mathcal{I}|s))^{-1}$. The theorem follows by applying inequality \eqref{inequality_zero_inconclusive} to this second sequential method.

Finally, we conclude this section by discussing the difficulties involved in improving the lower bound using this proof technique. In particular, small improvements to the one-shot lower bound given in Corollary \ref{corollary_probability_of_error}, such as  
%-----------
\begin{equation}
    \mathbb{P}_{\mathrm{error}}(\{\sigma_i\}_{i=1}^N)
    \geq 
    \frac{1}{g(N)} \left(1 - \frac{1}{2} \min_{i \neq j} \|\sigma_i - \sigma_j\|_1 \right),
\end{equation}
%-----------
where $g(N) < N$ is some increasing function, would unfortunately yield the same lower bound when taking 
$A = \log_K g(N)^{1+\delta}$ and following the same procedure. Therefore, improving Theorem \ref{theorem_general_lower_bound} would likely require exploring alternative proof techniques or refining Lemma \ref{ref_sequential_to_fixed}.

\subsection{Lower bound in quantum hypothesis testing}

Interestingly, the methodology used in the previous section can be applied to the case where a sequential method is used for the standard hypothesis testing task, i.e., there is no inconclusive answer, and the probability of error may be nonzero.
%------------
\begin{theorem}\label{theorem_probability_of_error}
    For any sequential method that discriminates the set of pure states $\{\ket{\psi_i}\}_{i=1}^N$ with a probability of error satisfying $\max_{s\in [N]}\mathbb{P}(\mathcal{E}|s) \leq \frac{3}{7}$, then
    %--------------
    \begin{equation}\label{probability_of_error_theorem}
        \frac{1}{-7\log \max_{i,j: i\neq j} |\bra{\psi_i}\ket{\psi_j}|^{2}} \left( \frac{1}{2}- \frac{7\max_{s\in [N]}\mathbb{P}(\mathcal{E}|s) }{6}\right)^2 \leq \left \lceil \max_{s\in[N]} \mathbb{E}[L|s]\right \rceil
    \end{equation}
\end{theorem}
%-------------
As in Theorem \ref{theorem_general_lower_bound}, we observe that the lower bound depends inversely on the minimum pairwise Chernoff distance of the different states. Additionally, similar to the dependency of the bound in Theorem \ref{theorem_general_lower_bound} on the inconclusive probability, here, as the probability of error increases, the bound decreases. This theorem complements the results in \cite{martinez2021quantum}, which show that any sequential method solving a binary hypothesis testing problem for mixed states $\{\sigma_1, \sigma_2\}$ satisfies
%------------
\begin{equation}
    \max_{s\in [2]}\mathbb{E}[L|s]\geq \frac{-\log \mathbb{P}(\mathcal{E})(1-\mathbb{P}(\mathcal{E}))}{\min\{D(\sigma_1,\sigma_2),D(\sigma_2,\sigma_1)\}}+O(1)
\end{equation}
%---------
However, this expression does not hold when $\min\{D(\rho_1,\rho_2),D(\rho_2,\rho_1)\}$ is not finite, as is the case for pure states.

Note that the bounds given in \eqref{theorem_lower_bound_inconclusive_prop} and \eqref{probability_of_error_theorem} do not coincide when $\mathbb{P}(\mathcal{E}) = \mathbb{P}(\mathcal{I}) = 0$, as they are derived through slightly different procedures. Specifically, the proof of this theorem is analogous to that of Theorem \ref{theorem_general_lower_bound}, but it uses Lemma \ref{ref_sequential_to_fixed_probability_of_error} instead of Lemma \ref{ref_sequential_to_fixed}, which is stated below.
%-------------
\begin{lemma}\label{ref_sequential_to_fixed_probability_of_error}
    For any sequential method that distinguishes a set of pure states $ \{ \ket{\psi_i} \}_{i=1}^N $ with a probability of error $\mathbb{P}(\mathcal{E}|s)<\frac{K-1}{2K}$ for all $s\in [N]$, using an expected number of measurements $\mathbb{E}[L|s] $, there exists at least one fixed-length method that uses $ AK \left \lceil \max_{s\in [N]}\mathbb{E}[L|s] \right \rceil$ samples and distinguishes the states with a probability of error at most 
    \begin{equation*}
        2e^{-A\left( \frac{1}{2}- \frac{K\max_{s\in [N]}\mathbb{P}(\mathcal{E}|s) }{K-1} \right)^2}
    \end{equation*}
    where $A,K\in \mathbb{N}$ and $K\geq 7$.
\end{lemma}

\begin{proof}
    See Appendix \ref{appendix_C}.
\end{proof}

As before, we use the fact that the specific fixed-length method employed in Lemma \ref{ref_sequential_to_fixed_probability_of_error} has an error probability that is lower bounded by that of the optimal fixed-length method. This, together with Corollary \ref{corollary_probability_of_error}, yields
%-------------
\begin{align}
    2e^{-A\left( \frac{1}{2}- \frac{K\max_{s\in [N]}\mathbb{P}(\mathcal{E}|s) }{K-1} \right)^2} &\geq  \mathbb{P}_{\mathrm{error}}\left(\left\{\ket{\psi_i}\bra{\psi_i}^{\otimes AK \left \lceil \max_{s\in[N]} \mathbb{E}[L|s]\right \rceil} \right\}_{i=1}^N \right) \nonumber \\ &  \geq \frac{1}{N} \left(1-\sqrt{1- \max_{i\neq j} |\bra{\psi_i}\ket{\psi_j}|^{2AK\left \lceil \max_{s\in [N]}\mathbb{E}[L|s] \right \rceil}} \right)
\end{align}
%-------------
for $K,A \in \mathbb{N}$, and $K\geq 7$. To obtain the result shown in Theorem \ref{theorem_probability_of_error}, we need to substitute 
%-------------
\begin{equation}
    A=\frac{\log 2 N^{1+\delta}}{\left( \frac{1}{2}- \frac{K\max_{s\in [N]}\mathbb{P}(\mathcal{E}|s)}{K-1} \right)^2}
\end{equation}
%-----------
into the previous expression. Consequently,
%----------
\begin{equation}
    \frac{1}{K\log \max_{i\neq j} |\bra{\psi_i}\ket{\psi_j}|^{2}} \frac{\log \left(1-\left(1-\frac{1}{N^\delta}\right)^2\right)}{\log(2N^{1+\delta})}\left( \frac{1}{2}- \frac{K\max_{s\in [N]}\mathbb{P}(\mathcal{E}|s)}{K-1} \right)^2 \leq \left \lceil \max_{s\in [N]}\mathbb{E}[L|s] \right \rceil
\end{equation}
%--------------
Since the inequality holds for any integer $ K \geq 7$ and  $\delta \in \mathbb{R}^+ $, we choose $K = 7$ and let $\delta \to \infty$ as they maximize the lower bound. Next, for a fixed $N$,
%------------
\begin{equation}
    \lim_{\delta \to \infty} \frac{\log \left( 1 - \left( 1 - \frac{1}{N^\delta} \right)^2 \right)}{\log(2N^{1 + \delta})} = -1,
\end{equation}
%--------------
which leads to the desired result.

\section{Upper bound in unambiguous discrimination}
\label{method_introduction}

In this section, we show an upper bound on the expected number of copies required to unambiguously discriminate a set of pure states of the optimal sequential method. The result is presented in the following theorem.

\begin{theorem}\label{theorem_2}
    For any $p\in [0,1]$, the optimal sequential method that unambiguously discriminates the set of states $\{\ket{\psi_i}\}_{i=1}^N$, such that $\max_{s\in [N]}\mathbb{P}(\mathcal{I}|s)\leq p$, satisfies the following inequality:
    %---------------
    \begin{equation}
        \max_{s\in[N]}\mathbb{E}[L|s]\leq \min_{\delta \in (0,1)} \left\{\frac{\left(1-p\right)\log(\frac{N-1}{\delta})}{-\log \max_{i\neq j} |\bra{\psi_i}\ket{\psi_j}|} \left(\frac{1+\delta}{1-\delta} \right) \right\}
    \end{equation}
    %--------------  
\end{theorem}

For this theorem, we rely on the simple fact that any sequential method, in expectation, uses more copies than the optimal method. Specifically, this upper bound is established using a novel non-adaptive approach, which we now present. As before, we first derive the inequality for the case where the inconclusive probability is zero, and afterwards, through Lemma \ref{last_lemma}, we generalize it.

The proposed method measures copies of the state $\ket{\psi}$ until a conclusive result is obtained. It proceeds as follows. Each copy is first measured using a two-outcome measurement, and in case the first outcome is obtained, it is then measured with an $N+1$-outcome POVM. The first measurement uses POVM $\{\Lambda^\dagger \Lambda, I - \Lambda^\dagger \Lambda\}$, where matrix $\Lambda$ is selected such that

\begin{enumerate}[$(i)$]
\item $0\preceq	 \Lambda^\dagger\Lambda  \preceq I$, 
\item and $\{\Lambda\ket{\psi_i} \}_{i=1}^N$ is an orthogonal set of vectors. 
\end{enumerate}

Without loss of generality, we adopt the convention that if the first outcome is obtained when measuring the state $\ket{\psi}$, the state transforms into $\Lambda \ket{\psi}/\sqrt{\bra{\psi} \Lambda^\dagger \Lambda \ket{\psi}}$\footnote{A POVM does not uniquely determine the post-measurement state. However, the post-measurement state of different POVM implementations varies only by a unitary operation, which is irrelevant to our analysis.}. Consequently, if the outcome of the POVM corresponds to $\Lambda^\dagger\Lambda$, then the set of post-measurement states becomes an orthonormal set. Hence, the second POVM, given by
%-----------------
\begin{equation}
\left \{ \frac{\Lambda \ket{\psi_1}\bra{\psi_1}\Lambda^\dagger}{\bra{\psi_1}\Lambda^\dagger\Lambda \ket{\psi_1}},\frac{\Lambda \ket{\psi_2}\bra{\psi_2}\Lambda^\dagger}{\bra{\psi_2}\Lambda^\dagger\Lambda \ket{\psi_2}},\cdots,\frac{\Lambda \ket{\psi_N}\bra{\psi_N}\Lambda^\dagger}{\bra{\psi_N}\Lambda^\dagger\Lambda \ket{\psi_N}}, I-\sum_{i=1}^N \frac{\Lambda \ket{\psi_i}\bra{\psi_i}\Lambda^\dagger}{\bra{\psi_i}\Lambda^\dagger\Lambda \ket{\psi_i}}  \right\},
\end{equation}
%-----------------
identifies the post-measurement state resulting from the first measurement, thereby allowing the correct identification of the initial state $\ket{\psi}$.

Now that we have explained the method, let's propose a specific matrix $ \Lambda$ that satisfies conditions ($i$) and ($ii$) under certain assumptions.

\begin{proposition}\label{proposition_1}
    Iff $\rank(W)=N$, where $W = \sum_{i=1}^N \ket{\psi_i} \bra{\psi_i}$, i.e., when the states $\{\ket{\psi_i}\}_{i=1}^N$ are linearly independent, then matrix 
    %--------------
    \begin{equation}
        \Lambda=\frac{\sqrt{W^+}}{\left \| \sqrt{W^+} \right \|_2} 
    \end{equation} 
    %---------------
    where $\|M\|_2:=\sup_{u:\|u\|_2=1} \|Mu\|_2=\sqrt{\lambda_{\max}(M^\dagger M)}$ and $M^+$ denotes the pseudoinverse of $M$, satisfies that: 
    %-------------
    \begin{enumerate}[$(i)$]
    \item $0\preceq	 \Lambda^\dagger\Lambda  \preceq I$, 
    \item and $\{\Lambda\ket{\psi_i} \}_{i=1}^N$ is an orthogonal set of vectors. 
    \end{enumerate}
    %-------------
\end{proposition}

\begin{proof}

    Substituting $\Lambda = \frac{\sqrt{W^+}}{\left\| \sqrt{W^+} \right\|_2}$, we obtain 
    %-----------
    \begin{equation}
        0\preceq\Lambda^\dagger \Lambda = \frac{\sqrt{W^+}^\dagger \sqrt{W^+}}{\lambda_{\mathrm{max}}(\sqrt{W^+}^\dagger \sqrt{W^+})}= \frac{W^+}{\lambda_{\max}(W^+)}\preceq I 
    \end{equation}
    %_----------
    where the equality follows from the fact that $W$ is Hermitian. As for the second condition, we prove that $\bra{\psi_i} W^+ \ket{\psi_j} = \delta_{i,j}$ when $\rank(W) = N$. To do this, we first define the matrix
    %-----------------
    \begin{equation}
        H = \left[\ket{\psi_1}, \ket{\psi_2}, \dots, \ket{\psi_N}\right].
    \end{equation}
    %----------
    Matrix $W$ can be written as $H H^\dagger$, and additionally, matrix $M_{i,j}=\bra{\psi_i} W^+ \ket{\psi_j}$ is given by $H^\dagger W^{+} H$. Since $N=\rank(W)=\rank(H)$, then $H^+ H=I_N$, where $I_N$ is the $N\times N$ identity matrix. Therefore, 
    %----------
    \begin{align}
        M&=H^\dagger W^+ H=H^\dagger (H H^\dagger)^+ H \nonumber \\ &= H^\dagger H^{\dagger \,+} H^+ H= ( H^+ H)^\dagger  H^+ H =I_N
    \end{align}
    %-----------
    where the third equality uses that $(A A^\dagger)^+=A^{\dagger\,+} A^+$ for any matrix $A$.

    Finally, we prove that the condition $\rank(W) = N$ is a necessary condition. To see this, we first use the fact that any $\Lambda$ that satisfies the conditions $(i)$ and ($ii$) also fulfills
    %--------------
    \begin{equation}
    	\Lambda \ket{\psi_i} \bra{\psi_i} \Lambda^\dagger=\alpha_i \ket{v_i}\bra{v_i}
    \end{equation}
    %-------------
    where $\{\ket{v_i}\}_{i=1}^N$ is a set of $N$ orthonormal quantum states. Hence,
    %--------------
    \begin{equation}
    	 \Lambda \left(\sum_{i=1}^N \ket{\psi_i} \bra{\psi_i} \right) \Lambda^\dagger= \sum_{i=1}^N \alpha_i \ket{v_i}\bra{v_i}
    \end{equation}
    %----------------
    which implies that $\rank(\Lambda W \Lambda^\dagger)=N$. Furthermore, using the property
    %--------------
    \begin{equation}
        \rank(AB)\leq \min\{\rank(A),\rank(B)\}.
    \end{equation}
    %--------------
    it follows that $N=\rank(\Lambda W \Lambda^\dagger)\leq \rank(W)$. Since $W$ is the sum of $N$ rank-1 matrices, we have that $\rank(W)\leq N$. Therefore, $\rank(W)=N$. Interestingly, the condition $\rank(W) = N$ is not only necessary for $\Lambda = \frac{\sqrt{W^+}}{\left\| \sqrt{W^+} \right\|_2}$ to satisfy conditions $(i)$ and ($ii$), but also necessary for the existence of any matrix $\Lambda$ that satisfies both conditions.

\end{proof}

\subsection{Performance of the proposed method}

In this section, we analyze the performance of the method when using the $\Lambda$ matrix presented in Proposition \ref{proposition_1}, which assumes linear independence of the states $\{\ket{\psi_i}\}_{i=1}^N$. Specifically, as in the rest of the paper, the performance is measured based on the expected number of copies used.

\begin{lemma}\label{lemma_4}
     The maximum expected number of copies used, $\max_{s\in [N]}\mathbb{E}[L|s]$, when applying the method introduced in Section \ref{method_introduction}, is upper bounded as
    %---------------
    \begin{equation}
    	\max_{s\in [N]}\mathbb{E}[L|s] \leq \frac{\lambda_{\mathrm{max}}(W)}{\lambda_{\mathrm{min}}^+(W)}
    \end{equation}
    %-----------
    where $W = \sum_{i=1}^N \ket{\psi_i} \bra{\psi_i}$, and $\lambda_{\mathrm{min}}^+(W)$ denotes the smallest positive eigenvalue of $W$.
\end{lemma}

\begin{proof}

Denoting by $\ket{\psi_s}$ the observed quantum state, then the probability of observing the first outcome with POVM $\{\Lambda^\dagger \Lambda, I-\Lambda^\dagger \Lambda\}$ is $\bra{\psi_s}\Lambda^\dagger\Lambda \ket{\psi_s}$. Therefore, the expected number of copies is given by
%---------------
\begin{equation}
	\mathbb{E}[L|s]=\sum_{l=0}^\infty l  (1-\bra{\psi_s}\Lambda^\dagger\Lambda \ket{\psi_s})^{l-1}\bra{\psi_s}\Lambda^\dagger\Lambda \ket{\psi_s}=\frac{1}{\bra{\psi_s}\Lambda^\dagger\Lambda \ket{\psi_s}}
\end{equation}
%--------------
where the second  equality uses that $\sum_{i=0}^\infty i \,\alpha^i=\alpha/(1-\alpha)^2$ for $|\alpha|<1$. Next, note that $\ket{\psi_s}\perp \text{Ker}(W)$, which leads to the following inequality,
%------------------
\begin{align}
	 \bra{\psi_s} \Lambda^\dagger \Lambda \ket{\psi_s} & \geq  \lambda_{\mathrm{min}}^+(\Lambda^\dagger \Lambda) \nonumber \\& = \lambda_{\min}^+ \left( \frac{W^+}{\lambda_{\max}(W^+)} \right) \nonumber \\ & = \frac{\lambda_{\min}^+ (W^+)}{\lambda_{\max}(W^+)}= \frac{\lambda_{\min}^+(W)}{\lambda_{\max}(W)}
\end{align}
%------------------
where $\lambda_{\mathrm{min}}^+(\Lambda^\dagger \Lambda)$ denotes the smallest positive eigenvalue of $\Lambda^\dagger\Lambda$. The first equality follows from $\Lambda=\sqrt{W^+}/\left \| \sqrt{W^+} \right \|_2$ and the last equality uses that $\lambda_{\max}(W^+)=1/\lambda_{\min}^+(W)$ and $\lambda_{\min}^+(W)=1/\lambda_{\max}(W^+)$. Therefore,
%-------------
\begin{equation}
	 \mathbb{E}[L|s]=\frac{1}{\bra{\psi_s} \Lambda^\dagger \Lambda \ket{\psi_s}}\leq \frac{\lambda_{\mathrm{max}}(W)}{\lambda_{\mathrm{min}}^+(W)}
\end{equation}
%-------------

\end{proof}

Interestingly, as indicated by Lemma \ref{lemma_4}, when the states are nearly orthogonal, the required number of copies is quite small. This is because, in such cases, all non-zero eigenvalues of $W$ are close to 1. To illustrate this more clearly, let us analyze an example, where the set of states is chosen such that they satisfy $ \| \ket{\psi_i} \bra{\psi_i} - \ket{i} \bra{i} \|_2 = \epsilon$. Specifically, we set $ \epsilon = \frac{1}{2N^2}$. Notably, for any value of $\epsilon$, the method introduced in \cite{perez2022quantum} requires $O(N)$ copies. In contrast, the last method requires fewer than 3 samples on average. To prove this result, we begin by using
    %--------------
    \begin{equation}
       1-N^2 \epsilon \leq \lambda_{\min}^+(W) \leq \lambda_{\max}(W)\leq 1+N \epsilon
    \end{equation}
    %--------
    The bound on $\lambda_{\max}(W)$ follows from
    %--------------
    \begin{align}
        \lambda_{\max}(W)&=\| W\|_2\leq \left \|W- \sum_{i=1}^N \ket{i}\bra{i} \right\|_2+ \left \|\sum_{i=1}^N \ket{i}\bra{i} \right\|_2 \nonumber  \\ & \leq N \epsilon +1
    \end{align}
    %----------------------
    For the other bound, we use
    %-----------------------
    \begin{equation}
        N=\Tr(W)\leq \lambda_{\min}^+(W)+(N-1) \lambda_{\max}(W)\leq  \lambda_{\min}^+(W) +(N-1)(1+N\epsilon )
    \end{equation}
    %-----------------
    Therefore, performing some algebraic operations,
    %--------------
    \begin{equation}
        1-N^2 \epsilon \leq \lambda_{\min}^+ (W)
    \end{equation}
    %----------
    Consequently,
    %-------------
    \begin{equation}
        \mathbb{E}[L] \leq \frac{1+N \epsilon}{1-N^2\epsilon}\leq 3
    \end{equation}
    %-------------
    where the second inequality follows from substituting $\epsilon$. %_------------
    
\subsection{Removing the requirement for linear independence}\label{removing_linear_independence}

The method presented at the beginning of Section \ref{method_introduction} is based on Proposition \ref{proposition_1} and, therefore, requires the set of states to be linearly independent. However, a sequential method that unambiguously discriminates a set of states does not need to satisfy this condition. In general, the only necessary requirement is that $|\bra{\psi_i} \ket{\psi_j}| < 1 $ for all $ i \neq j $. To achieve this, we can modify the proposed method to collectively measure a set of $k$ copies, instead of measuring each copy individually. That is, the problem becomes one of distinguishing the states $ \{ \ket{\psi_i}^{\otimes k} \}_{i=1}^N $. In this case, we can apply the previous method as long as the states $ \ket{\psi_i}^{\otimes k} $ are linearly independent. We will now show that a small value of $k$ is generally sufficient to ensure the linear independence of the states. Furthermore, we show also that the ratio between the largest and smallest positive eigenvalues of $ W^{(k)} = \sum_{i} \ket{\psi_i} \bra{\psi_i}^{\otimes k} $ is small.

To do this, first, let us define matrix $H^{(k)}=[\ket{\psi_1}^{\otimes k},\ket{\psi_2}^{\otimes k},\cdots,\ket{\psi_N}^{\otimes k}]$. Importantly,  $H^{(k)^{\dagger}} H^{(k)}\rightarrow I_N$ as $k$ increases since $[H^{(k)^{\dagger}} H^{(k)}]_{i,j}=\bra{\psi_i}\ket{\psi_j}^{k} $. Therefore, using Gershgorin's circle theorem, the eigenvalues of $H^{(k)^{\dagger}} H^{(k)}$ satisfy
%-----------------
\begin{equation}\label{inequality_eigenvalues}
    1-(N-1) \max_{i\neq j} |\bra{\psi_i}\ket{\psi_j}|^{k}\leq \lambda\leq 1+(N-1) \max_{i\neq j} |\bra{\psi_i}\ket{\psi_j}|^{k}
\end{equation}
%----------------
Therefore, for any $\delta\in (0,1)$, if $k$ satisfies
%-------------
\begin{equation}
    k\geq \frac{\log(\frac{N-1}{\delta})}{-\log \max_{i\neq j} |\bra{\psi_i}\ket{\psi_j}|}
\end{equation}
%---------
then the eigenvalues belong to the interval $[1-\delta,1+\delta]$. Hence, since all eigenvalues are positive, $\rank(H^{(k)^{\dagger}} H^{(k)})=N$, which is equal to the $\rank(H^{(k)})$. That is, the states are linearly independent. Furthermore, since the singular values satisfy $\sigma(H^{(k)})=\sigma(H^{(k)\dagger})$, implying that for the non-zero eigenvalues $\lambda(H^{(k)^{\dagger}} H^{(k)})=\lambda(W^{(k)})$. That is, the non-zero eigenvalues of $W^{(k)}$ fulfill inequality \eqref{inequality_eigenvalues}. Hence, putting everything together, and using the value of $k$ specified earlier, we conclude that
%-----------
\begin{equation}\label{final_equation}
    \mathbb{E}[L|s]\leq \min_{\delta \in (0,1)} \left\{\frac{\log(\frac{N-1}{\delta})}{-\log \max_{i\neq j} |\bra{\psi_i}\ket{\psi_j}|} \left(\frac{1+\delta}{1-\delta} \right) \right\}
\end{equation}
%-------------
for any set of $N$ quantum states. Therefore, the modification explained in this section improves efficiency in terms of the expected number of copies, but at the cost of introducing the added complexity of performing collective measurements, which are extremely difficult to implement. Additionally, note that equation \eqref{final_equation} proves Theorem \ref{theorem_2} for $p=0$. The result can be generalized for $p\in [0,1]$ using the following lemma:
%----------------
\begin{lemma}\label{last_lemma}
    For any $p\in [0,1]$, and any sequential method that unambiguously discriminates a set of states $\{\ket{\psi_i}\}_{i=1}^N$ with an inconclusive probability $\mathbb{P}(\mathcal{I})=0$ and an expected number of copies $\mathbb{E}[L|s]$, there exists another sequential method that unambiguously discriminates the same set of states $\{\ket{\psi_i}\}_{i=1}^N$ with an inconclusive probability $\mathbb{P}(\mathcal{I}|s) = p$ and an expected number of copies equal to $\mathbb{E}[L|s](1-\mathbb{P}(\mathcal{I}|s))$.
\end{lemma}
%----------------
\begin{proof}
    The idea of the proof is to modify the sequential method slightly to introduce the inconclusive probability. In particular, before anything else, we sample a Bernoulli random variable $B$ with parameter $p$. If the outcome is $1$, then the method terminates and outputs $D = N + 1$. Otherwise, the original sequential method is run. From this, it follows that $\mathbb{P}(\mathcal{I} | s)=p$ for all $s\in[N]$. Let's now move on to analyze the expected value of copies. Clearly, $E[L|s,B=1]=0$, and similarly, $E[L|s,B=0]=\mathbb{E}[L|s]$. Consequently, the expected value of the new sequential method is given by $E[L|s](1-\mathbb{P}(\mathcal{I}|s))$.
\end{proof}
%----------------
Interestingly, by examining the results of Theorem \ref{theorem_general_lower_bound} and Theorem \ref{theorem_2}, we observe that the gap between the two bounds is $O(\log N)$. As a result, the optimal method can only improve the performance of the proposed method by a factor of this order. In other words, while our proposed method may not be optimal, it achieves performance that is very close to optimal.

\section{Analysis for Mixed States}

In this section, we discuss generalizations of some of the main results from the previous sections to the case of mixed states. Interestingly, in Theorem \ref{theorem_general_lower_bound}, we observe that neither Lemma \ref{ref_sequential_to_fixed} nor Lemma \ref{lemma_3} makes use, at any point, of the restriction to pure states, therefore, they also apply to mixed states. Next, combining \eqref{eq_probability_of_error} with inequality
%----------
\begin{equation}
    \|\rho-\sigma\|_1\leq 2 \sqrt{1-F(\rho,\sigma)}
\end{equation}
%-----------
yields 
%------------
\begin{equation}
    \mathbb{P}_{\mathrm{error}}(\{\sigma_i\}_{i=1}^N)\geq \frac{1}{N} \left(1-\sqrt{1- \max_{i\neq j}F(\sigma_i, \sigma_j)} \right)
\end{equation}
%------------
Consequently, using that $F(\rho^{\otimes K}, \sigma^{\otimes K}) = F(\rho, \sigma)^K$, we obtain an expression analogous to \eqref{expression_p_error}. Hence, an analogous lower bound for mixed states holds.
%-----------
\begin{proposition}\label{generalization_theorem_general_lower_bound}
    For any sequential method that distinguishes unambiguously a set of mixed states $\{\sigma_i\}_{i=1}^N$, the maximum expected value of state copies $\max_{s\in[N]}\mathbb{E}[L|s]$ satisfies
    %---------------
    \begin{equation}
        \frac{1-\max_{s\in [N]} \mathbb{P}(\mathcal{I}|s)}{-\log \max_{i,j: i\neq j} F(\sigma_i,\sigma_j)}\left( \frac{\log 3}{3 }\right) \leq \left \lceil \max_{s\in[N]} \mathbb{E}[L|s]\right \rceil
    \end{equation}
    %---------------
\end{proposition}
%-----------
Note that, importantly, unlike in the pure case, for some sets of mixed states $\{\sigma_i\}_{i=1}^N$, all sequential methods that distinguish the states unambiguously satisfy $\max_{s \in [N]} \mathbb{P}(\mathcal{I} | s) = 1$. In other words, no sequential method can discriminate these states unambiguously. The necessary and sufficient condition for the existence of a sequential method that satisfies $\max_{s \in [N]} \mathbb{P}(\mathcal{I} | s) < 1$ is stated next.

\begin{proposition}
    The necessary and sufficient condition for the existence of a sequential method that unambiguously discriminates a set of mixed states $\{\sigma_i\}_{i=1}^N$, such that $\max_{s \in [N]} \mathbb{P}(\mathcal{I} \mid s) < 1$, is $\mathrm{supp}(\sigma_i) \nsubseteq \mathrm{supp}(\sigma_j)$ for all $i \neq j$.
\end{proposition}

\begin{proof}
    The sufficiency is shown in Section 4C of \cite{perez2022quantum}. To show that the previous condition is necessary, we use a proof by contradiction. Specifically, assume that there exists a set of states $\{\sigma_i\}_{i=1}^N$ such that the condition $\mathrm{supp}(\sigma_i) \nsubseteq \mathrm{supp}(\sigma_j)$ for all $i \neq j$ is not satisfied, and that there exists a sequential method $A$ such that $\max_{s\in [N]} \mathbb{P}(\mathcal{I}|s)<1$. Then, there exist $i,j \in [N]$ such that $\mathrm{supp}(\sigma_i) \subseteq \mathrm{supp}(\sigma_j)$. Let us denote these states as $\rho_1$ and $\rho_2$. From \cite{martinez2021quantum}, we know that any sequential method that discriminates these two states must satisfy
    %---------
    \begin{equation}\label{lower_bound_}
        \max_{s\in [2]}\mathbb{E}[L|s]\geq \frac{-\log \mathbb{P}(\mathcal{E})(1-\mathbb{P}(\mathcal{E}))}{\min\{D(\rho_1,\rho_2),D(\rho_2,\rho_1)\}}+O(1)
    \end{equation}
    %---------
    Importantly, since $\mathrm{supp}(\rho_1) \subseteq \mathrm{supp}(\rho_2)$, $\min\{D(\rho_1,\rho_2), D(\rho_2,\rho_1)\}$ is finite, and the expected value grows at least as $O(-\log \mathbb{P}(\mathcal{E}))$. However, using Lemma \ref{lemma_3}, the sequential method $A$ for unambiguous state discrimination can be transformed into a sequential method with $\mathbb{P}(\mathcal{I})=0$, or equivalently, $\mathbb{P}(\mathcal{E})=0$. Therefore, this transformed method, when applied to the states $\{\rho_1, \rho_2\}$, also satisfies $\mathbb{P}(\mathcal{E})=0$ and has a bounded expected value. This leads to a contradiction with \eqref{lower_bound_} and completes the proof.
\end{proof}

Finally, the method described in Section \ref{method_introduction} can also be generalized to apply to some sets of mixed states. The procedure starts by measuring with a POVM $\{\Lambda^\dagger \Lambda, I - \Lambda^\dagger \Lambda\}$, where observing the first outcome maps the supports of the states to orthogonal subspaces. Subsequently, a POVM of projectors onto these subspaces is performed, with an additional element projecting onto the complement to guarantee that the sum of all POVM elements equals the identity.

Specifically, a matrix $\Lambda$ satisfying the previous conditions exists at least when the states $\{\ket{v^{i}_j}\}_{i,j}$, formed by the eigenvectors of the different states $\{\sigma_i\}$, are linearly independent, where $\ket{v^{i}_j}$ denotes the $j^{th}$ eigenvector of state $\sigma_i$. In this scenario, we can take 
%-----------
\begin{equation}
    \Lambda = \frac{\sqrt{W^+}}{\left\| \sqrt{W^+} \right\|_2},
\end{equation}
%------------
where $W = \sum_{i=1}^N \sum_{j=1}^{d_i} \ket{v^{i}_j}\bra{v^{i}_j}$.

In summary, the lower bound presented in Theorem \ref{theorem_general_lower_bound} can be fully generalized to mixed states, as shown in Proposition \ref{generalization_theorem_general_lower_bound}. Similarly, the method proposed in Section \ref{method_introduction} can also be generalized to mixed states. However, the upper bound derived from this method is not straightforward to obtain for the mixed-state scenario.

\section{Conclusions}

In summary, by combining the results from both Theorem \ref{theorem_general_lower_bound} and Theorem \ref{theorem_2}, we have demonstrated that the performance of the optimal sequential method for unambiguous discrimination of a set of pure states, as measured by $\max_{s\in[N]}\mathbb{E}[L|s]$, is primarily determined by the inverse of the minimum pairwise Chernoff distance. Furthermore, we have shown that the dependence of $\max_{s\in[N]}\mathbb{E}[L|s]$ on the number of states, $N$, is at most $O(\log N)$. Despite these findings, an important question remains unresolved: namely, what is the optimal scaling of the expected number of samples as a function of the number of states, $N$? That is, the optimal method may have no dependency on the number of states $N$ at all, which would show that the lower bound of Theorem \ref{theorem_general_lower_bound} is tight. Alternatively, the lower bound could be improved to demonstrate that the $O(\log N)$ factor is necessary, rather than merely an artifact of the proof of Theorem \ref{theorem_2}.

Apart from these results on optimal sequential schemes, we have also introduced a novel non-adaptive method for the unambiguous discrimination of $N$ pure states and a generalization of Theorem \ref{theorem_general_lower_bound} for mixed states. This method achieves a highly competitive performance,  as the optimal strategy can improve its performance only by at most a factor of $O(\log N)$. Furthermore, the method is quite simple and has an intuitive geometric interpretation.

\section{Acknowledgments}

This work has been funded by grants PID2022-137099NB-C41, PID2019-104958RB-C41 funded by MCIN/AEI/10.13039/501100011033 and FSE+ and by grant 2021 SGR 01033 funded by AGAUR, Dept. de Recerca i Universitats de la Generalitat de Catalunya 10.13039/501100002809.

\bibliographystyle{quantum}
\bibliography{ref.bib}

@article{bergou2012optimal,
  title={Optimal unambiguous discrimination of pure quantum states},
  author={Bergou, J{\'a}nos A and Futschik, Ulrike and Feldman, Edgar},
  journal={Physical review letters},
  volume={108},
  number={25},
  pages={250502},
  year={2012},
  publisher={APS},
  doi={10.1103/PhysRevLett.108.250502}
}

@article{chefles1998unambiguous,
  title={Unambiguous discrimination between linearly independent quantum states},
  author={Chefles, Anthony},
  journal={Physics Letters A},
  volume={239},
  number={6},
  pages={339--347},
  year={1998},
  publisher={Elsevier},
  doi={10.1016/S0375-9601(98)00064-4}
}

@article{sentis2022online,
  title={Online identification of symmetric pure states},
  author={Sent{\'\i}s, Gael and Mart{\'\i}nez-Vargas, Esteban and Mu{\~n}oz-Tapia, Ramon},
  journal={Quantum},
  volume={6},
  pages={658},
  year={2022},
  publisher={Verein zur F{\"o}rderung des Open Access Publizierens in den Quantenwissenschaften},
  doi={10.22331/q-2022-02-21-658}
}

@article{peres1988differentiate,
  title={How to differentiate between non-orthogonal states},
  author={Peres, Asher},
  journal={Physics Letters A},
  volume={128},
  number={1-2},
  pages={19},
  year={1988},
  publisher={Elsevier},
  doi={10.1016/0375-9601(88)91034-1}
}

@article{dieks1988overlap,
  title={Overlap and distinguishability of quantum states},
  author={Dieks, Dennis},
  journal={Physics Letters A},
  volume={126},
  number={5-6},
  pages={303--306},
  year={1988},
  publisher={Elsevier},
  doi={10.1016/0375-9601(88)90840-7}
}

@article{ivanovic1987differentiate,
  title={How to differentiate between non-orthogonal states},
  author={Ivanovic, Igor D},
  journal={Physics Letters A},
  volume={123},
  number={6},
  pages={257--259},
  year={1987},
  publisher={Elsevier},
  doi={10.1016/0375-9601(87)90222-2}
}

@article{jaeger1995optimal,
  title={Optimal distinction between two non-orthogonal quantum states},
  author={Jaeger, Gregg and Shimony, Abner},
  journal={Physics Letters A},
  volume={197},
  number={2},
  pages={83--87},
  year={1995},
  publisher={Elsevier},
  doi={10.1016/0375-9601(94)00919-G}
}

@article{herzog2005optimum,
  title={Optimum unambiguous discrimination of two mixed quantum states},
  author={Herzog, Ulrike and Bergou, J{\'a}nos A},
  journal={Physical Review A—Atomic, Molecular, and Optical Physics},
  volume={71},
  number={5},
  pages={050301},
  year={2005},
  publisher={APS},
  doi={10.1103/PhysRevA.71.050301}
}

@book{Helstrom,
    author    = "C. W. Helstrom",
    title     = "Quantum Detection and Estimation",
    year      = "1976",
    publisher = "Academic Press",
    address   = "New York",
    doi       = "10.1007/BF01007479"
}

@article{Chernoff_2,
  title = {The Chernoff lower bound for symmetric quantum hypothesis testing},
  author = {M. Nussbaum and A. Szko\l{}a},
  journal = {The Annals of Statistics},
  volume = {37},
  issue = {2},
  pages = {1040–1057},
  year = {2009},
  doi = {10.1214/08-AOS593}
}

@article{Chernoff_1,
  title = {Discriminating States: The Quantum Chernoff Bound},
  author = {Audenaert, K. M. R. and Calsamiglia, J. and Mu\~noz-Tapia, R. and Bagan, E. and Masanes, Ll. and Acin, A. and Verstraete, F.},
  journal = {Phys. Rev. Lett.},
  volume = {98},
  issue = {16},
  pages = {160501},
  numpages = {4},
  year = {2007},
  month = {Apr},
  publisher = {American Physical Society},
  doi = {10.1103/PhysRevLett.98.160501},
  url = {https://link.aps.org/doi/10.1103/PhysRevLett.98.160501}
}

@article{Multiple,
    author =       "Ke Li",
    title =        "Discriminating quantum states: The multiple Chernoff distance",
    journal =      "Ann. Statist. 44(4): 1661-1679 ",
    year={2016},
    doi={10.1214/16-AOS1436}
}

@article{perez2022quantum,
  title={Quantum Multiple Hypothesis Testing Based on a Sequential Discarding Scheme},
  author={P{\'e}rez-Guijarro, Jordi and Pag{\`e}s-Zamora, Alba and Fonollosa, Javier Rodr{\'\i}guez},
  journal={IEEE access},
  volume={10},
  pages={13813--13826},
  year={2022},
  publisher={IEEE},
  doi={10.1109/ACCESS.2022.3143706}
}

@article{casa_2,
  title = {Quantum Chernoff bound as a measure of distinguishability between density matrices: Application to qubit and Gaussian states},
  author = {Calsamiglia, J. and Mu\~noz-Tapia, R. and Masanes, Ll. and Acin, A. and Bagan, E.},
  journal = {Phys. Rev. A},
  volume = {77},
  issue = {3},
  pages = {032311},
  numpages = {15},
  year = {2008},
  month = {Mar},
  publisher = {American Physical Society},
  doi = {10.1103/PhysRevA.77.032311},
  url = {https://link.aps.org/doi/10.1103/PhysRevA.77.032311}
}

@article{gisin2007quantum,
  title={Quantum communication},
  author={Gisin, Nicolas and Thew, Rob},
  journal={Nature photonics},
  volume={1},
  number={3},
  pages={165--171},
  year={2007},
  publisher={Nature Publishing Group UK London},
  doi={10.1038/nphoton.2007.22}
}

@article{acin2006secrecy,
  title={Secrecy properties of quantum channels},
  author={Ac{\'\i}n, Antonio and Bae, Joonwoo and Bagan, E and Baig, M and Masanes, Ll and Mu{\~n}oz-Tapia, Ramon},
  journal={Physical Review A—Atomic, Molecular, and Optical Physics},
  volume={73},
  number={1},
  pages={012327},
  year={2006},
  publisher={APS},
  doi={10.1103/PhysRevA.73.012327}
}

@article{slussarenko2017quantum,
  title={Quantum state discrimination using the minimum average number of copies},
  author={Slussarenko, Sergei and Weston, Morgan M and Li, Jun-Gang and Campbell, Nicholas and Wiseman, Howard M and Pryde, Geoff J},
  journal={Physical review letters},
  volume={118},
  number={3},
  pages={030502},
  year={2017},
  publisher={APS},
  doi={10.1103/PhysRevLett.118.030502}
}

@article{takagi2022fundamental,
  title={Fundamental limits of quantum error mitigation},
  author={Takagi, Ryuji and Endo, Suguru and Minagawa, Shintaro and Gu, Mile},
  journal={npj Quantum Information},
  volume={8},
  number={1},
  pages={114},
  year={2022},
  publisher={Nature Publishing Group UK London},
  doi={10.1038/s41534-022-00618-z}
}

@article{li2022optimal,
  title={Optimal adaptive strategies for sequential quantum hypothesis testing},
  author={Li, Yonglong and Tan, Vincent YF and Tomamichel, Marco},
  journal={Communications in Mathematical Physics},
  volume={392},
  number={3},
  pages={993--1027},
  year={2022},
  publisher={Springer},
  doi={10.1007/s00220-022-04362-5}
}

@article{martinez2021quantum,
  title={Quantum sequential hypothesis testing},
  author={Mart{\'\i}nez Vargas, Esteban and Hirche, Christoph and Sent{\'\i}s, Gael and Skotiniotis, Michalis and Carrizo, Marta and Mu{\~n}oz-Tapia, Ramon and Calsamiglia, John},
  journal={Physical review letters},
  volume={126},
  number={18},
  pages={180502},
  year={2021},
  publisher={APS},
  doi={10.1103/PhysRevLett.126.180502}
}

\begin{appendices}

    \section{Proof of Lemma \ref{ref_sequential_to_fixed}}\label{Appendix_B}

    In this appendix, we provide the proof of Lemma \ref{ref_sequential_to_fixed}, which, for completeness, is restated here.

    \begin{lemma*}
        For any sequential method that unambiguously distinguishes a set of pure states $ \{ \ket{\psi_i} \}_{i=1}^N $ with $\mathbb{P}(\mathcal{I})=0$ using a maximum expected number of measurements $\max_{s\in [N]} \mathbb{E}[L|s] $, there exists at least one fixed-length method that uses $ AK  \left \lceil \max_{s\in[N]} \mathbb{E}[L|s]\right \rceil $ samples and distinguishes the states with a probability of error at most $ 1/K^A $, where $A,K\in \mathbb{N}$.
    \end{lemma*}

    \begin{proof}

    The fixed-length method is defined as follows:
    %-------------
    \begin{itemize}
        \item Divide the $ A K  \left \lceil \max_{s\in[N]} \mathbb{E}[L|s]\right \rceil $ samples into $ A $ groups.
        \item Run the sequential method independently for each group.
        \item The final decision coincides with the one given by any sequential method that has completed the procedure. If the sequential method does not finish for any group, a random decision is made.
    \end{itemize}
    %---------------
    To analyze this method, first, let $L_s$ denote the number of copies used by the sequential method when the true state is given by $\ket{\psi_s}$. Then, 
    %-------------
    \begin{equation}
        \mathbb{P}(L_s \geq K \mathbb{E}[L|s]) \leq \frac{1}{K}
    \end{equation}
    %---------------
    independently of the method used. This follows from Markov's inequality. Therefore,  
    %---------------
    \begin{equation}\label{markov_inequality_number_of_samples}
        \mathbb{P}\left(L_s > K \left \lceil \max_{s\in[N]} \mathbb{E}[L|s]\right \rceil\right) \leq \frac{1}{K}
    \end{equation}
    %--------------
    holds for all $s\in [N]$. Hence, 
    %--------------
    \begin{align}
        \mathbb{P}(\mathcal{E})&=\mathbb{P}(D\neq S)=\sum_{s=1}^N \pi_s \,\mathbb{P} ( D\neq s |S=s) \nonumber \\& \leq  \sum_{s=1}^N \pi_s \,\mathbb{P}\left(L_s > K \left \lceil \max_{s\in[N]} \mathbb{E}[L|s]\right \rceil\right)^A \nonumber \\ &\leq \frac{1}{K^A}
    \end{align}
    %---------------
    where the first inequality follows from the fact that an error can only occur if none of the $A$ groups has finished. This happens with probability $\mathbb{P} \left( L_s > K \lceil \max_{s\in[N]} \mathbb{E}[L|s] \rceil\right)^A$, as each group uses $K \lceil \max_{s\in[N]} \mathbb{E}[L|s] \rceil$ states. The last inequality uses \eqref{markov_inequality_number_of_samples}.

        \end{proof}

        \section{Proof of Lemma \ref{lemma_3}}\label{appendix_A}

    In this section we prove the following result.

    \begin{lemma**}
        For any sequential method that unambiguously discriminates a set of states $\{\ket{\psi_i}\}_{i=1}^N$ with an inconclusive probability $\mathbb{P}(\mathcal{I}|s)<1 $ and an expected number of copies $\mathbb{E}[L|s]$, there exists another sequential method that unambiguously discriminates the same set of states $\{\ket{\psi_i}\}_{i=1}^N$ with an inconclusive probability $\mathbb{P}(\mathcal{I}|s) = 0$ and an expected number of copies equal to $\mathbb{E}[L|s]/(1-\mathbb{P}(\mathcal{I}|s))$.
    \end{lemma**}

    \begin{proof}
        The sequential method with zero inconclusive probability is obtained by running the other sequential method until a conclusive answer is obtained. The number of times the method is run is denoted by the r.v. $R$. Importantly,
        %-----------------
        \begin{equation}
            \mathbb{P}(R=r|s)= \mathbb{P}(\mathcal{I}|s)^{r-1} (1-\mathbb{P}(\mathcal{I}|s))
        \end{equation}
        %-----------------
        Furthermore, for $R=r$, the expected number of copies used is 
        %-------------
        \begin{equation}
            (r-1)\mathbb{E}[L|D=N+1,s]+\mathbb{E}[L|D<N+1,s]
        \end{equation}
        %--------------
        where $D=N+1$ is the inconclusive event. Therefore, the expected number of copies used by the new method is
        %------------
        \begin{align}
            &\sum_{r=1}^{\infty} \mathbb{P}(R=r|s) \left( (r-1)\mathbb{E}[L|D=N+1,s]+\mathbb{E}[L|D<N+1,s] \right) \nonumber \\&=\sum_{r=1}^{\infty} \mathbb{P}(\mathcal{I}|s)^{r-1} (1-\mathbb{P}(\mathcal{I}|s)) \left( (r-1)\mathbb{E}[L|D=N+1,s]+\mathbb{E}[L|D<N+1,s] \right) \nonumber 
        \end{align}
        \begin{align}
            & = \mathbb{E}[L|D<N+1,s]  \sum_{r=1}^{\infty} \mathbb{P}(\mathcal{I}|s)^{r-1} (1-\mathbb{P}(\mathcal{I}|s)) \nonumber \\&\hspace{4cm}+\mathbb{E}[L|D=N+1,s]  \sum_{r=1}^{\infty} \mathbb{P}(\mathcal{I}|s)^{r-1} (1-\mathbb{P}(\mathcal{I}|s)) (r-1) \nonumber \\&= \mathbb{E}[L|D<N+1,s]+\mathbb{E}[L|D=N+1,s] \frac{\mathbb{P}(\mathcal{I}|s)}{1-\mathbb{P}(\mathcal{I}|s)} \nonumber \\&= \frac{(1-\mathbb{P}(\mathcal{I}|s)) \mathbb{E}[L|D<N+1,s]+\mathbb{P}(\mathcal{I}|s) \mathbb{E}[L|D=N+1,s]  }{1-\mathbb{P}(\mathcal{I}|s)} \nonumber \\&= \frac{\mathbb{E}[L|s]}{1-\mathbb{P}(\mathcal{I}|s)}
        \end{align}
        %-----------
        where the third equality uses that $\sum_{i=0}^\infty \alpha^i=1/(1-\alpha)$ and $\sum_{i=0}^\infty i \,\alpha^i=\alpha/(1-\alpha)^2$ for $|\alpha|<1$. The last equality uses the identities $1-\mathbb{P}(\mathcal{I}|s)=P(D<N+1|s)$, and $\mathbb{P}(\mathcal{I}|s)=\mathbb{P}(D=N+1|s)$.

    \end{proof}

    \section{Proof of Lemma \ref{ref_sequential_to_fixed_probability_of_error}}\label{appendix_C}

    In this appendix, we provide the proof of Lemma \ref{ref_sequential_to_fixed_probability_of_error}, which, for completeness, is restated here.

    \begin{lemma***}
        For any sequential method that distinguishes a set of pure states $ \{ \ket{\psi_i} \}_{i=1}^N $ with a probability of error $\mathbb{P}(\mathcal{E}|s)<\frac{K-1}{2K}$ for all $s\in [N]$, using an expected number of measurements $\mathbb{E}[L|s] $, there exists at least one fixed-length method that uses $ AK \left \lceil \max_{s\in [N]}\mathbb{E}[L|s] \right \rceil$ samples and distinguishes the states with a probability of error at most 
        \begin{equation*}
            2e^{-A\left( \frac{1}{2}- \frac{K\max_{s\in [N]}\mathbb{P}(\mathcal{E}|s) }{K-1} \right)^2}
        \end{equation*}
        where $A,K\in \mathbb{N}$ and $K\geq 7$.
    \end{lemma***}

    \begin{proof}
    The method coincides with the one explained in Lemma \ref{ref_sequential_to_fixed}, except in the decision taken. In particular, we decide by a majority vote between the segments that have finished. In case no segment ends the procedure, the decision is taken at random. Now, let's analyze its performance. To do so, first note that
    %--------------
    \begin{align}
        \mathbb{P} (\mathcal{E}|s)&=\mathbb{P}\left(\mathcal{E}|L_s> K \left \lceil \max_{s\in [N]}\mathbb{E}[L|s]\right \rceil,s\right) \mathbb{P}\left(L_s> K \left \lceil \max_{s\in [N]}\mathbb{E}[L|s]\right \rceil \right) \nonumber \\&  +\mathbb{P} \left(\mathcal{E}|L_s \leq K \left \lceil \max_{s\in [N]}\mathbb{E}[L|s]\right \rceil,s \right) \mathbb{P}\left(L_s \leq K \left \lceil \max_{s\in [N]}\mathbb{E}[L|s]\right \rceil \right) \nonumber \\ & \geq \left(1-\frac{1}{K}\right) \mathbb{P}\left(\mathcal{E}|L_s \leq  K \left \lceil \max_{s\in [N]}\mathbb{E}[L|s]\right \rceil,s\right)
    \end{align}
    %------------
    where the inequality uses \eqref{markov_inequality_number_of_samples}. Therefore, $\mathbb{P}(\mathcal{E}|L_s\leq K \left \lceil \max_{s\in [N]}\mathbb{E}[L|s]\right \rceil,s)\leq \frac{ K}{K-1} \mathbb{P}(\mathcal{E}|s) \leq \frac{ K}{K-1} \max_{s \in [N]}\mathbb{P}(\mathcal{E}|s)$. Now, we move on to the probability of error associated with the majority vote among the methods that have finished. For the moment, let's assume that the number of segments that have finished is $\tilde{A} \leq A$, and for these segments, we define the following random variables,
    %---------------
    \begin{equation}
        X_i^{(s)} =\mathbbm{1} \left\{ \text{Segment $i$ makes the wrong decision given that the correct decision is $s$} \right\}
    \end{equation}
    %--------------
    Hence, since in these segments $L_s \leq K \left \lceil \max_{s\in [N]}\mathbb{E}[L|s]\right \rceil$,
    %-----------------------------
    \begin{equation}
        \mathbb{E}\left[X_i^{(s)}\right] = \mathbb{P}\left(\mathcal{E}|L_s \leq  K \left \lceil \max_{s\in [N]}\mathbb{E}[L|s]\right \rceil,s \right)
    \end{equation}
    %_----------
    Using Hoeffding's inequality, 
    %-------------
    \begin{equation}
        \mathbb{P}\left( \sum_{i=1}^{\tilde{A}}X_i^{(s)} -\mathbb{P}\left(\mathcal{E}|L_s \leq K \left \lceil \max_{s\in [N]}\mathbb{E}[L|s]\right \rceil,s \right) \tilde{A}\geq t \right)\leq e^{-\frac{2t^2}{\tilde{A}}}
    \end{equation}
    %-------------
    Taking $t= \left( 1/2-\mathbb{P}(\mathcal{E}|L_s \leq K\left \lceil \max_{s\in [N]}\mathbb{E}[L|s]\right \rceil,s) \right) \tilde{A}$ (as $t$ has to be positive, this implies that to guarantee $t\geq 0$, $\max_{s \in [N]}\mathbb{P}(\mathcal{E}|s) \leq \frac{K-1}{2K}$) the previous inequality becomes, 
    %-------------
    \begin{equation}
        \mathbb{P}\left( \sum_{i=1}^{\tilde{A}}X_i^{(s)} \geq  \frac{\tilde{A}}{2}\right)\leq e^{-2\tilde{A} \left(\frac{1}{2}-\frac{K\max_{s \in [N]}\mathbb{P}(\mathcal{E}|s) }{K-1}\right)^2}
    \end{equation}
    %-----------
    The left-hand side of the previous inequality is an upper bound on the probability of error of the majority vote. To see this, note that if $\sum_{i=1}^{\tilde{A}} X_i^{(s)} < \tilde{A}/2$, then the majority of the groups that have finished are correct, and therefore, the decision is correct. Therefore, errors can only happen when $\sum_{i=1}^{\tilde{A}} X_i^{(s)} \geq \tilde{A}/2$. Consequently,
    %-----------
    \begin{equation}
        \mathbb{P}(\mathcal{E}_{\mathrm{fixed}}|s,\tilde{A}) \leq  \mathbb{P}\left( \sum_{i=1}^{\tilde{A}}X_i^{(s)} \geq  \frac{\tilde{A}}{2}\right) \leq e^{-2\tilde{A} \left(\frac{1}{2}-\frac{K\max_{s \in [N]}\mathbb{P}(\mathcal{E}|s) }{K-1}\right)^2}
    \end{equation}
    %-----------
    where $\mathcal{E}_{\mathrm{fixed}}$ denotes the error event of the fixed-length method. Now, let's analyze $\tilde{A}$. First, we denote
    %---------
    \begin{equation}
        Y_i^{(s)} = \mathbbm{1} \left\{ \text{Segment $i$ does not terminate} \right\}
    \end{equation}
    %---------
    Hence, $\mathbb{E}[Y_i^{(s)}]=\mathbb{P}(L_s > K \left \lceil \max_{s\in [N]}\mathbb{E}[L|s]\right \rceil)$. Using Hoeffding's inequality,
    %-----------
    \begin{equation}
        \mathbb{P}\left( \sum_{i=1}^A Y_i^{(s)}- A\mathbb{P}\left(L_s > K \left \lceil \max_{s\in [N]}\mathbb{E}[L|s]\right \rceil \right) \geq t \right) \leq e^{\frac{-2t^2}{A}}
    \end{equation}
    %-----------
    Taking $t=A\left(\frac{1}{2}- \mathbb{P}(L_s > K \left \lceil \max_{s\in [N]}\mathbb{E}[L|s]\right \rceil ) \right)$ (this implies that $K\geq 2$ to guarantee that $t\geq 0$),
    %-----------
    \begin{equation}
        \mathbb{P}\left( \sum_{i=1}^A Y_i^{(s)}\geq \frac{A}{2} \right) \leq e^{-2A \left( \frac{1}{2}- \frac{1}{K}\right)^2}
    \end{equation}
    %-----------
    Therefore, 
    %-----------
    \begin{equation}
        \mathbb{P}\left( \tilde{A}\leq \frac{A}{2} \,| \,s\right) \leq e^{-2A \left( \frac{1}{2}- \frac{1}{K}\right)^2}
    \end{equation}
    %------------
    Putting everything together, we have that 
    %-----------
    \begin{align}
        \mathbb{P}(\mathcal{E}_{\mathrm{fixed}})&=\sum_{s,\tilde{A}} \mathbb{P}(\mathcal{E}_{\mathrm{fixed}}|s,\tilde{A}) \mathbb{P}(\tilde{A},s)  \nonumber \\ &= \sum_{s,\tilde{A}} \mathbb{P}(\mathcal{E}_{\mathrm{fixed}}|s,\tilde{A}) \mathbb{P}(\tilde{A}|s) \,\pi_s \nonumber \\ &= \sum_{s,\tilde{A}>\frac{A}{2}} \mathbb{P}(\mathcal{E}_{\mathrm{fixed}}|s,\tilde{A}) \mathbb{P}(\tilde{A}|s)\,\pi_s+\sum_{s,\tilde{A}\leq \frac{A}{2}} \mathbb{P}(\mathcal{E}_{\mathrm{fixed}}|s,\tilde{A}) \mathbb{P}(\tilde{A}|s)\pi_s \nonumber  \\ & \leq e^{-A \left(\frac{1}{2}-\frac{K\max_{s \in [N]}\mathbb{P}(\mathcal{E}|s) }{K-1}\right)^2} + e^{-2A \left( \frac{1}{2}- \frac{1}{K}\right)^2}
    \end{align}
    %-----------
    Finally, for $K\geq 7$, as 
    %-----------
    \begin{equation}
        \left(\frac{1}{\sqrt{2}}-\frac{\sqrt{2}}{K}\right)\geq \frac{1}{2} \geq \left(\frac{1}{2}-\frac{K\max_{s \in [N]}\mathbb{P}(\mathcal{E}|s) }{K-1}\right) \geq 0
    \end{equation}
    %-----------
    it follows that
    %-----------
    \begin{equation}
        \mathbb{P}(\mathcal{E}_{\mathrm{fixed}})\leq 2 e^{-A \left(\frac{1}{2}-\frac{K\max_{s \in [N]}\mathbb{P}(\mathcal{E}|s) }{K-1}\right)^2}
    \end{equation}
    %_-----------

\end{proof}

\end{appendices}

\end{document}